\documentclass[preprint]{elsarticle}

\usepackage[english]{babel}
\usepackage{enumerate}
\usepackage{color}
\usepackage{subfigure}
\usepackage{dsfont}
\usepackage{graphicx}
\usepackage{epstopdf}
\usepackage{amsmath}
\usepackage{mathtools}
\usepackage{amsthm}
\usepackage{amstext}
\usepackage{amssymb}
\usepackage{mathrsfs}
\usepackage{mathtools}
\usepackage{tikz}
\usetikzlibrary{arrows,positioning}

\usepackage{algorithm}
\usepackage{algorithmic}

\def\R{\mathbb{R}}

\def\Lip{\mathop{\rm Lip}}

\def\argmin{\mathop{\rm arg\, min}}

\def\P{{\mathcal P}}

\def\hQ{{\hat Q}}
\def\hmu{{\hat \mu}}

\def\sX{{\mathsf X}}

\def\sA{{\mathsf A}}

\def\sE{{\mathsf E}}

\theoremstyle{remark}

\newtheorem{definition}{Definition}
\newtheorem{theorem}{Theorem}

\newtheorem{lemma}{Lemma}
\theoremstyle{remark}
\newtheorem{remark}{Remark}
\newtheorem{assumption}{Assumption}

\allowdisplaybreaks

\usepackage{amsopn}

\begin{document}

\begin{frontmatter}
\sloppy
\title{Value Iteration Algorithm for Mean-field Games}

\author[BCN]{Berkay Anahtarci, Can Deha Kariksiz, Naci Saldi}%\ead{berkay.anahtarci@ozyegin.edu.tr,candeha@gmail.com,naci.saldi@ozyegin.edu.tr}     

\address[BCN]{\"{O}zye\u{g}in University, \c{C}ekmek\"{o}y, \.{I}stanbul, Turkey, Emails:berkay.anahtarci@ozyegin.edu.tr,candeha@gmail.com,naci.saldi@ozyegin.edu.tr}  
       
\begin{abstract}
In the literature, existence of mean-field equilibria has been established for discrete-time mean field games under both the discounted cost and the average cost optimality criteria. In this paper, we provide a value iteration algorithm to compute  mean-field equilibrium for both the discounted cost and the average cost criteria, whose existence proved previously. We establish that the value iteration algorithm converges to the fixed point of a mean-field equilibrium operator. Then, using this fixed point, we construct a mean-field equilibrium. In our value iteration algorithm, we use $Q$-functions instead of value functions.  
\end{abstract}


\end{frontmatter}

\section{Introduction}\label{sec1}

In this paper, we propose a value iteration algorithm to compute an equilibrium solution for discrete-time Polish state mean-field games under both the discounted cost and the average cost optimality criteria. In the literature, the existence of mean-field equilibria has been established for a very general class of mean-field game models under both the discounted cost \cite{SaBaRaSIAM} and the average cost \cite{Wie19} optimality criteria. Here, using $Q$-functions, we develop a value iteration algorithm via the \emph{so-called} mean-field equilibrium (MFE) operator. This operator is very similar to the Bellman optimality operator in classical stochastic control problems. The (only) difference is that the MFE operator also updates the distribution of the state, in addition to the $Q$-function, in each iteration. We first establish that the MFE operator is a contraction under some regularity assumptions on the system components. Therefore, by Banach Fixed Point Theorem, there exists a fixed point of this operator. Then, we construct a mean-field equilibrium using the fixed point of MFE operator. Here, we use $Q$-functions instead of value functions, which are generally used in classical value iteration algorithms.

Mean-field games are the infinite population limits of finite-agent stochastic games with mean-field interactions. Establishing the existence of the Nash equilibrium for a finite-agent game problem is in general infeasible due to the decentralized nature of the information structure and the large number of coupled agents. To overcome these difficulties, one possible way is to consider the infinite-population limit of the problem and use the fact that the equilibrium solution of the infinite-population limit is approximately Nash in a finite-agent setting when the number of agents are sufficiently large. Note that, in the limiting case, a generic agent is faced with a single-agent stochastic control problem with a constraint on the distribution of the state at each time (i.e., a mean-field game problem). The equilibrium solution in the infinite-population limit is a pair which contains a policy and a state measure. This pair should satisfy Nash certainty equivalence (NCE) principle \cite{HuMaCa06} which states that, under a given state measure, the policy should be optimal and when the generic agent applies this policy, the resulting distribution of the agent's state is same as the state measure. The purpose of this paper is to develop a value iteration algorithm for computing such an equilibrium under the discounted cost and the average cost criteria.

Mean-field games have been introduced by Huang, Malham\'{e}, and Caines \cite{HuMaCa06} and Lasry and Lions \cite{LaLi07} to establish the existence of approximate Nash equilibria for continuous-time differential games with a large number of agents interacting through a mean-field term. In continuous-time differential games, mean-field equilibrium can be obtained by solving a Fokker-Planck (FP) equation evolving forward in time and a Hamilton-Jacobi-Bellman (HJB) equation evolving backward in time. We refer the reader to
 \cite{HuCaMa07,TeZhBa14,Hua10,BeFrPh13,Ca11,CaDe13,GoSa14,MoBa16} for studies of continuous-time mean-field games with different models and cost functions, such as games with major-minor players, risk-sensitive games, games with Markov jump parameters, and LQG games.

In continuous-time mean-field games, there is an extensive literature on numerical computation of mean-field equilibrium. In \cite{AcCa10}, authors develop finite-difference method to FP and HJB partial differential equations when the state space is two-dimensional torus. In \cite{AcCaCa13}, convergence of finite difference method for coupled FP and HJB equations is studied when Hamiltonian has a special structure and the state space is again two dimensional torus. Reference \cite{AcPo16} considers convergence of finite difference schemes to the weak solutions of coupled FP and HJB equations arising in mean field games. In \cite{AlFeGo16}, authors develop two numerical methods (variational and monotonic) for stationary mean field games when the state space is one dimensional torus. Reference \cite{GoSa18} studies continuous-time finite-state mean field games and develops numerical methods for such games satisfying monotonicity condition. In \cite{Gue12}, author establishes numerical methods for mean field games with quadratic costs. We refer the reader to the survey paper \cite{AcLa20} for comprehensive literature review on numerical aspects of continuous-time mean field games.

Although there is an extensive literature on continuous-time mean-field games, this is not so in the discrete-time setting. Existing studies mostly consider discrete (finite or countable) state games or linear games or games where the mean-field term only affects the cost functions; that is, the evolution of the states of the agents is independent. Reference \cite{GoMoSo10} considers a discrete-time mean-field game with a finite state space over a finite horizon. In \cite{AdJoWe15}, discrete-time mean-field game with countable state-space is studied subject to an infinite-horizon discounted cost criterion. References \cite{ElLiNi13,MoBa15,NoNa13,MoBa16-cdc} consider discrete-time mean-field games with linear state dynamics. Reference \cite{SaBaRaSIAM} considers a discrete-time mean-field game with Polish state and action spaces under the discounted cost optimality criteria. There are only three papers \cite{Bis15,Wie19,WiAl05} studying discrete-time mean-field games subject to the average cost optimality criteria. In \cite{WiAl05}, the authors consider a discrete set-up for average-cost mean-field games. In \cite{Bis15}, the author considers average-cost mean-field games with Polish state spaces. In that paper, it was assumed that, for the finite agent game problem, the dynamics of the agents do not depend on the mean-field term. Under strong conditions on system components, \cite{Bis15} establishes the existence of  Nash equilibria for finite-agent games, and then, shows that these Nash equilibria converge to the mean-field equilibria in the infinite-population limit. Reference \cite{Wie19} considers average-cost mean-field games with compact state spaces. 

We note that the aforementioned papers, except linear models, mostly identify the existence of mean-field equilibrium and no algorithm with convergence guarantee has been proposed to compute this mean-field equilibrium in these works. The only work that establish the computation of mean-field equilibrium in discrete-time setup is \cite{HaSi19}. In this work, to compute mean-field equilibrium, authors develop fictitious play iteration for the finite-state mean-field games under the finite-horizon cost criterion. The convergence of the proposed algorithm is established under monotonicity condition, which is in general imposed to ensure the uniqueness of mean-field equilibrium and is quite restrictive. The studies that consider abstract state spaces (non-discrete) have only established the existence of mean-field equilibrium and no algorithm with convergence guarantee has been proposed to compute this mean-field equilibrium. Our work appears to be the first one that studies this problem for mean-field games with abstract state spaces. Additionally, in this work, we both consider infinite-horizon discounted cost and average cost optimality criteria, which have not been studied previously in the discrete-time mean-field game literature for computational purposes. Finally, it is known that analysis of continuous-time and discrete-time setups are quite different, requiring different set of tools. Therefore, it is unlikely to apply methods reviewed above for continuous-time setup to discrete-time case. Moreover, numerical methods developed for continuous-time mean-field games in general assume special state form like one or two dimensional torus. Here, we consider mean-field games with arbitrary Polish state spaces.

The paper is organized as follows. In Section~\ref{sec3}, we introduce the infinite population mean-field game and define the mean-field equilibrium. In Section~\ref{sec2}, we formulate the finite-agent version of the game problem. In Section~\ref{main-proof}, we propose the value iteration algorithm for both the discounted cost and the average cost optimality criteria. In Section~\ref{discount} we prove the convergence of the value iteration algorithm for discounted cost. In Section~\ref{average} we prove the convergence of the value iteration algorithm for average cost. Section~\ref{conc} concludes the paper.

\smallskip

\noindent\textbf{Notation.} 
For a metric space $\sE$, we let $\P(\sE)$ denote the set of all Borel probability measures on $\sE$. A sequence $\{\mu_n\}$ of measures on $\sE$ is said to converge weakly to a measure $\mu$ if $\int_{\sE} g(e) \mu_n(de)\rightarrow\int_{\sE} g(e) \mu(de)$ for all $g: \sE \rightarrow \R$ that are bounded and continuous. The set of probability measures $\P(\sE)$ is endowed with the Borel $\sigma$-algebra induced by weak convergence. The notation $v\sim \nu$ means that the random element $v$ has distribution $\nu$. Unless  specified otherwise, the term ``measurable" will refer to Borel measurability. 

\section{Mean-field games and mean-field equilibria}\label{sec3}

A discrete-time mean-field game is specified by
\begin{align}
\bigl( \sX, \sA, p, c, \mu_0 \bigr), \nonumber
\end{align}
where $\sX$ and $\sA$ are the state and the action spaces, respectively. Here, $\sX$ is a Polish space (complete separable metric space) with the metric $d_{\sX}$ and $\sA$ is a compact subset of a finite dimensional Euclidean space $\R^d$ with the Euclidean distance norm $\|\cdot\|$. The measurable function $p : \sX \times \sA \times \P(\sX) \to \P(\sX)$ denotes the transition probability of the next state given the previous state-action pair and the state-measure. The measurable function $c: \sX \times \sA \times \P(\sX) \rightarrow [0,\infty)$ is the one-stage cost function. The measure $\mu_0$ is the initial state distribution.

In this model, a policy $\pi$ is a stochastic kernel on $\sA$ given $\sX$; that is, $\pi:\sX \rightarrow \P(\sA)$ is a measurable function. Let $\Pi$ denote the set of all policies. By the Ionescu Tulcea Theorem \cite{HeLa96}, a policy $\pi$ and an initial measure $\mu_0$ define a unique probability measure $P^{\pi}$ on $(\sX \times \sA)^{\infty}$. The expectation with respect to $P^{\pi}$ is denoted by $E^{\pi}$.

It is important to note that a mean-field game is neither a game nor a stochastic control problem in the strict sense. We have a single agent that tries to minimize an objective function as in stochastic control problems, but this agent should also compete with the constraint on the state distribution at each time step as in game problems. More precisely, we have a single agent and we model the collective behavior of (a large population of) other agents by an exogenous \textit{state-measure} $\mu \in \P(\sX)$. This measure $\mu$ should also be aligned with the state distribution of this single agent when the agent applies its optimal policy. The precise mathematical description of the problem is given as follows.

Let us fix a state-measure $\mu \in \P(\sX)$ that describes the collective behavior of the other agents. A policy $\pi^{*} \in \Pi$ is optimal for $\mu$ if
\begin{align}
W_{\mu}(\pi^{*}) = \inf_{\pi \in \Pi} W_{\mu}(\pi), \nonumber
\end{align}
where $W \in \{J,V\}$ and
\begin{align}
J_{\mu}(\pi) &= E^{\pi}\biggl[ \sum_{t=0}^{\infty} \beta^t c(x(t),a(t),\mu) \biggr], \nonumber \\
V_{\mu}(\pi) &= \limsup_{T \rightarrow \infty} \frac{1}{T} E^{\pi}\biggl[ \sum_{t=0}^{T-1} c(x(t),a(t),\mu) \biggr], \nonumber 
\end{align}
are the discounted cost and the average cost of policy $\pi$ under the state-measure $\mu$, respectively. Here, $\beta \in (0,1)$ is the discount factor. In this model, the evolution of the states and actions is given by
\begin{align}
x(0) &\sim \mu_0, \,\,\,
x(t) \sim p(\,\cdot\,|x(t-1),a(t-1),\mu), \text{ } t\geq1, \nonumber \\
a(t) &\sim \pi(\,\cdot\,|x(t)), \text{ } t\geq0. \nonumber
\end{align}
Now, define the set-valued mapping $\Psi : \P(\sX) \rightarrow 2^{\Pi}$  as 
$$\Psi(\mu) = \{\pi \in \Pi: \pi \text{ is optimal for }  \mu \text{ }\text{ and } \text{ } \mu_0 = \mu\};$$
that is, given $\mu$, the set $\Psi(\mu)$ is the set of optimal policies for $\mu$ when the initial distribution is $\mu$ as well.  

Conversely, we define another set-valued mapping $\Lambda : \Pi \to 2^{\P(\sX)}$ as follows: given $\pi \in \Pi$, the state-measure $\mu_{\pi}$ is in $\Lambda(\pi)$ if it is a fixed point of the equation
\begin{align}
\mu_{\pi}(\,\cdot\,) = \int_{\sX \times \sA} p(\,\cdot\,|x,a,\mu_{\pi})  \, \pi(da|x) \, \mu_{\pi}(dx). \nonumber
\end{align}
In other words, $\mu_{\pi}$ is the invariant distribution of the Markov transition probability $P(\cdot|x) = \int_{\sA} p(\,\cdot\,|x,a,\mu_{\pi}) \, \pi(da|x)$. Without any assumptions on the transition probability, it is possible to have $\Lambda(\pi) = \emptyset$ for some $\pi$. However, under Assumption~\ref{as1} below, $\Lambda(\pi)$ has an unique element for all $\pi$. Therefore, it is indeed a single-valued mapping. 

The notion of equilibrium for mean-field games is defined via these mappings $\Psi$, $\Lambda$ as follows.

\begin{definition}
A pair $(\pi_*,\mu_*) \in \Pi \times \P(\sX)$ is a \emph{mean-field equilibrium} if $\pi_* \in \Psi(\mu_*)$ and $\mu_* \in \Lambda(\pi_*)$. In other words, $\pi_*$ is an optimal policy given the state-measure $\mu_*$ and $\mu_*$ is the state distribution under the policy $\pi_*$. 
\end{definition}

In the literature, the existence of mean-field equilibria has been established for both the discounted cost \cite{SaBaRaSIAM} and the average cost \cite{Wie19}. In this paper, our goal is to develop a value iteration algorithm for computing a mean-field equilibrium. To that end, we will impose certain assumptions on the components of the mean-field game model. Before doing this, we need to give some definitions.

For any measurable function $u: \sX \times \sA \rightarrow \R$, let $u_{\min}(x) \coloneqq \inf_{a \in \sA} u(x,a)$ and $u_{\max}(x) \coloneqq \sup_{a \in \sA} u(x,a)$. Let $w:\sX \times \sA \rightarrow [1,\infty)$ be a continuous weight function. For any measurable $v: \sX \times \sA \rightarrow \R$, we define $w$-norm of $v$ as
$$
\|v\|_w \coloneqq \sup_{x,a} \frac{|v(x,a)|}{w(x,a)}. 
$$ 
For any measurable $u: \sX \rightarrow \R$, we define $w_{\max}$-norm of $u$ as
\begin{align}
\|v\|_{w_{\max}} &\coloneqq \sup_{x} \frac{|u(x)|}{w_{\max}(x)}. \nonumber
\end{align}
Let $B(\sX,K)$ be the set of real-valued measurable functions with $w_{\max}$-norm less than $K$. Let $C(\sX)$ be the set of real-valued continuous functions on $\sX$. For each $g \in C(\sX)$, let
\begin{align}
\|g\|_{\Lip} \coloneqq \sup_{(x,y)\in\sX\times\sX} \frac{|g(x)-g(y)|}{d_{\sX}(x,y)}. \nonumber
\end{align}
If $\|g\|_{\Lip}$ is finite, then $g$ is called Lipschitz continuous with Lipschitz constant $\|g\|_{\Lip}$. $\Lip(\sX)$ denotes the set of all Lipschitz continuous functions on $\sX$, i.e.,
\begin{align}
\Lip(\sX) \coloneqq \{g \in C(\sX): \|g\|_{\Lip} < \infty \} \nonumber
\end{align}
and $\Lip(\sX,K)$ denotes the set of all $g \in \Lip(\sX)$ with $\|g\|_{\Lip} \leq K$. For any $\mu, \nu \in \P(\sX)$, we denote by $C(\mu,\nu)$ the set of couplings between $\mu$ and $\nu$; that is, $\xi \in \P(\sX\times\sX)$ is an element of $C(\mu,\nu)$ if $\xi(\cdot \times \sX) = \mu(\cdot)$ and $\xi(\sX\times\cdot) = \nu(\cdot)$. The \emph{Wasserstein distance of order $1$} \cite[Definition 6.1]{Vil09} between two probability measures $\mu$ and $\nu$ over $\sX$ is defined as
$$
W_1(\mu,\nu) = \inf\left\{ \int_{\sX \times \sX} d_{\sX}(x,y) \, \xi(dx,dy): \xi \in C(\mu,\nu) \right\}. 
$$
By using Kantorovich-Rubinstein duality, we can also write Wasserstein distance of order $1$ \cite[p. 95]{Vil09} as follows:
\begin{align}
W_1(\mu,\nu) \coloneqq \sup \biggl\{\biggl|\int_{\sX} g d\mu - \int_{\sX} g d\nu\biggr|: g \in \Lip(\sX,1)\biggr\}. \nonumber
\end{align}
For compact $\sX$, the Wasserstein distance of order $1$ metrizes the weak topology on the set of probability measures $\P(\sX)$ (see \cite[Corollary 6.13, p. 97]{Vil09}). However, in general, it is stronger than weak topology. 

Finally, we define $F: \sX \times \Lip(\sX) \times \P(\sX) \times \sA \rightarrow \R$ as
$$
F: \sX \times \Lip(\sX) \times \P(\sX) \times \sA \ni (x,v,\mu,a) \mapsto 
c(x,a,\mu) + \xi \int_{\sX} v(y) \, p(dy|x,a,\mu) \in \R,
$$
where $\xi = \beta$ if the objective function is the discounted cost and $\xi = 1$ if the objective function is the average cost. We may now state our assumptions.

\begin{assumption}
\label{as1}
\begin{itemize}
\item [ ]
\item [(a)] The one-stage cost function $c$ is continuous. Moreover, it satisfies the following Lipschitz bounds:  
\begin{align}
\|c(\cdot,\cdot,\mu) - c(\cdot,\cdot,\hat{\mu})\|_w &\leq L_1 \, W_1(\mu,\hat{\mu}), \text{ } \text{ } \forall \mu, \hat{\mu}, \nonumber \\
\sup_{(a,\mu) \in \sA \times \P(\sX)} |c(x,a,\mu) - c(\hat{x},a,\mu)| &\leq L_2 \, d_{\sX}(x,\hat{x}), \text{ } \text{ } \forall x, \hat{x}. \nonumber
\end{align}
\item [(b)] The stochastic kernel $p(\,\cdot\,|x,a,\mu)$ is weakly continuous in $(x,a,\mu)$. Moreover, it satisfies the following Lipschitz bounds:
\begin{align}
&\sup_{x \in \sX} W_1(p(\cdot|x,a,\mu) - p(\cdot|x,\hat{a},\hat{\mu})) \leq K_1 \, \left(\|a -\hat{a}\| + W_1(\mu,\hat{\mu})\right), \forall \mu, \hat{\mu},  \forall a, \hat{a}, \nonumber \\
&\sup_{\mu \in \P(\sX)} W_1(p(\cdot|x,a,\mu) - p(\cdot|\hat{x},\hat{a},\mu)) \leq K_2 \, \left( d_{\sX}(x,\hat{x}) + \|a-\hat{a}\| \right), \text{ } \text{ } \forall x, \hat{x}, \forall a, \hat{a}.\nonumber
\end{align}
\item [(c)] $\sA$ is convex.
\item [(d)] There exist nonnegative real numbers $M$ and $\alpha$ such that for each $(x,a,\mu) \in \sX \times \sA \times \P(\sX)$, we have
\begin{align}
c(x,a,\mu) &\leq M \, w(x,a), \nonumber  \\
\int_{\sX} w_{\max}(y) \, p(dy|x,a,\mu) &\leq \alpha \, w(x,a). \label{disc-drift}
\end{align}
\item [(e)] Let ${\cal F}$ be the set of non-negative functions in 
$$\Lip\left(\sX,\frac{L_2}{1-\xi \, K_2}\right) \bigcap B\left(\sX,\frac{M}{1-\xi \, \alpha}\right),$$ where $\xi = \beta$ if the objective function is the discounted cost and $\xi = 1$ if the objective function is the average cost. For any $v \in {\cal F}$, $\mu \in \P(\sX)$, and $x \in \sX$, $F(x,v,\mu,\cdot)$ is $\rho$-strongly convex; that is, $F(x,v,\mu,\cdot)$ is differentiable with the gradient $\nabla F(x,v,\mu,\cdot)$ and it satisfies 
$$
F(x,v,\mu,a)  \geq F(x,v,\mu,\hat{a}) + \nabla F(x,v,\mu,\hat{a})^T \cdot (a-\hat{a}) + \frac{\rho}{2} \, \|a-\hat{a}\|^2,  
$$ 
for some $\rho > 0$ and for all $a,\hat{a} \in \sA$. Moreover, the gradient $\nabla F(x,v,\mu,a): \sX \times {\cal F} \times \P(\sX) \times \sA \rightarrow \R^d$ satisfies the following Lipschitz bound:
\begin{align}
\hspace{-20pt}\sup_{a \in \sA} \|\nabla F(x,v,\mu,a) - \nabla F(\hat{x},\hat{v},\hat{\mu},a)\| \leq K_F \, \left(d_{\sX}(x,\hat{x})+\|v-\hat{v}\|_{w_{\max}} + W_1(\mu,\hat{\mu})\right), \nonumber  \nonumber
\end{align}
for every $x, \hat{x}, v, \hat{v}, \mu,$ and $\hat{\mu}$.
\end{itemize}
\end{assumption}

Note that condition (d) is a standard assumption in the study of stochastic control problems with unbounded one-stage cost functions \cite{HeLa99}. Conditions (a) and (b) are required in order to control the effect of the state-measure $\mu$ on the value functions through the one-stage cost function and the state transition probability. Condition (e) is imposed to control the effect of the state-measure $\mu$ on the optimal policy. Indeed, this condition is equivalent to the canonical assumption that guarantees Lipschitz continuity, with respect to unknown parameters, of the optimal solutions of the convex optimization problem  \cite[Theorem 4.51]{BoSh00}.

\subsection{Finite Player Game}\label{sec2}

The model introduced in the previous section is actually the infinite-population limit of the finite-population game model that we describe below.

In this model, we have $N$-agents with state space $\sX$ and action space $\sA$. For every $t \in \{0,1,2,\ldots\}$ and every $i \in \{1,2,\ldots,N\}$, let $x^N_i(t) \in \sX$ and $a^N_i(t) \in \sA$ denote the state and the action of Agent~$i$ at time $t$, and 
\begin{align}
e_t^{(N)}(\,\cdot\,) \coloneqq \frac{1}{N} \sum_{i=1}^N \delta_{x_i^N(t)}(\,\cdot\,) \in \P(\sX) \nonumber
\end{align}
denote the empirical distribution of the state configuration at time $t$, where $\delta_x\in\P(\sX)$ is the Dirac measure at $x$. The initial states $x^N_i(0)$ are independent and identically distributed according to $\mu_0$, and, for each $t \ge 0$, the next-states $(x^N_1(t+1),\ldots,x^N_N(t+1))$ are generated according to the probability distribution
\begin{align}
&\prod^N_{i=1} p\big(dx^N_i(t+1)\big|x^N_i(t),a^N_i(t),e^{(N)}_t\big). \nonumber %\label{eq:state_spec}
\end{align}
A \emph{policy} $\pi$ for a generic agent is a stochastic kernel on $\sA$ given $\sX$. The set of all policies for Agent~$i$ is denoted by $\Pi_i$.

Let ${\bf \Pi}^{(N)} = \prod_{i=1}^N \Pi_i$. By ${\boldsymbol \pi}^{(N)} \coloneqq (\pi^1,\ldots,\pi^N)$, $\pi^i \in \Pi_i$, we denote an $N$-tuple of policies for all the agents in the game. Under such an $N$-tuple of policies, the actions at each time $t \ge 0$ are generated according to the probability distribution
\begin{align}\label{eq:policy_spec}
\prod^N_{i=1} \pi^i_t\big(da^N_i(t)\big|x^N_i(t)\big).
\end{align}

For Agent~$i$, the discounted cost and the average cost under the initial distribution $\mu_0$ and the $N$-tuple of policies ${\boldsymbol \pi}^{(N)} \in {\bf \Pi}^{(N)}$ are respectively given by
\begin{align}
J_i^{(N)}({\boldsymbol \pi}^{(N)}) &= E^{{\boldsymbol \pi}^{(N)}}\biggl[\sum_{t=0}^{\infty}\beta^{t}c(x_{i}^N(t),a_{i}^N(t),e^{(N)}_t)\biggr], \nonumber \\
V_i^{(N)}({\boldsymbol \pi}^{(N)}) &= \limsup_{T \rightarrow \infty} \frac{1}{T} E^{{\boldsymbol \pi}^{(N)}}\biggl[\sum_{t=0}^{T-1} c(x_{i}^N(t),a_{i}^N(t),e^{(N)}_t)\biggr].\nonumber
\end{align}

Using these definitions, the Nash equilibrium is defined for this game model as follows.

\begin{definition}
A policy ${\boldsymbol \pi}^{(N*)}= (\pi^{1*},\ldots,\pi^{N*})$ constitutes a \emph{Nash equilibrium} if
\begin{align}
W_i^{(N)}({\boldsymbol \pi}^{(N*)}) = \inf_{\pi^i \in \Pi_i} W_i^{(N)}({\boldsymbol \pi}^{(N*)}_{-i},\pi^i) \nonumber
\end{align}
for each $i=1,\ldots,N$, where ${\boldsymbol \pi}^{(N*)}_{-i} \coloneqq (\pi^{j*})_{j\neq i}$ and $W \in \{J,V\}$.
\end{definition}

For this game model, it is in general prohibitive to even prove the existence of Nash equilibria due to the (almost) decentralized nature of the information structure of the problem and the large number of players. However, when the number of players are sufficiently large, one way to overcome this challenge is to introduce the infinite-population limit $N\rightarrow\infty$ of the game (i.e., mean-field game). In this limiting case, we can model the empirical distribution of the state configuration as an exogenous state-measure, which should be consistent with the distribution of a generic agent by the law of large numbers (i.e., mean-field equilibrium). Hence, in the limiting case, a generic agent is exactly faced with a mean-field game that is introduced in the preceding section. One can then prove that if each agent in the finite-agent $N$ game problem adopts the mean-field equilibrium policy, the resulting policy will be an approximate Nash equilibrium for all sufficiently large $N$ (see, \cite[Theorem 4.1]{SaBaRaSIAM}, \cite[Section 5]{Wie19}). Therefore, by studying the infinite-population limit, which is easier to handle, one can obtain an approximate Nash equilibrium for the original finite-agent game problem for which establishing the existence of an exact Nash equilibrium is very difficult.

%The motivation for studying mean-field game problems is coming from the challenges to establish the existence of Nash equilibria for this finite-population game problems. 

\section{Value Iteration Algorithm}\label{main-proof}

Note that, given any state-measure $\mu \in \P(\sX)$, the optimal control problem for the mean-field game reduces to finding an optimal policy for a Markov decision process (MDP).
Since an optimal value and an optimal policy can be computed via a value iteration algorithm in MDPs, it is possible to develop a similar value iteration algorithm for computing the mean-field equilibrium. In this section, we develop such algorithms for both the discounted cost and the average cost optimality criteria using $Q$-functions.

\subsection{Discounted Cost}\label{discount}

In this section, we first state the value iteration algorithm for computing mean-field equilibrium and then establish the convergence of this algorithm. To that end, in addition to Assumption~\ref{as1}, we assume the following.

\begin{assumption}\label{as2}
\begin{itemize}
\item [ ]
\item [(a)] We assume that 
$$\hspace{-25pt}k \coloneqq \max\left\{\beta\,\alpha + \frac{K_F}{\rho}  K_1 , L_1+\beta \frac{L_2}{1-\beta\,K_2} K_1 + \bigg(\frac{K_F}{\rho} +1\bigg) K_1+ K_2 + \frac{K_F}{\rho}\right\} < 1.$$ 
\end{itemize}
\end{assumption}

For any state-measure $\mu$, we define the optimal value function of the optimal stochastic control problem by
$$
J_{\mu}^*(x) \coloneqq \inf_{\pi \in \Pi} E^{\pi}\biggl[ \sum_{t=0}^{\infty} \beta^t c(x(t),a(t),\mu) \, \bigg| \, x(0) = x \biggr]. 
$$
The following characterization of optimal policies is a known result in the theory of Markov Decision Processes (see \cite[Chapter 4]{HeLa96} and \cite[Chapter 8]{HeLa99}). First of all, the optimal value function $J_{\mu}^*(x)$ is the unique fixed point of the Bellman optimality operator $T_{\mu}$, which is $\beta \alpha$-contractive with respect to $w_{\max}$-norm; that is,
$$
J_{\mu}^*(x) = \min_{a \in \sA} \bigg[c(x,a,\mu) + \beta \int_{\sX} J_{\mu}^*(y) \, p(dy|x,a,\mu) \bigg] \eqqcolon T_{\mu}J_{\mu}^*(x).
$$
Moreover, if the mapping $f^*: \sX \rightarrow \sA$ attains the minimum in equation above; that is,
\begin{align}
&\min_{a \in \sA} \bigg[c(x,a,\mu) + \beta \int_{\sX} J_{\mu}^*(y) \, p(dy|x,a,\mu) \bigg] \nonumber \\
&\phantom{xxxxxxxxxxxxxxxx}= c(x,f^*(x),\mu) + \beta \int_{\sX} J_{\mu}^*(y) \, p(dy|x,f^*(x),\mu), \label{optim}
\end{align}
then the policy $\pi^*(a|x) = \delta_{f^*(x)}(a)$ is optimal. In the classical value iteration algorithm, the idea is to compute $J_{\mu}^*(x)$ by iteratively applying the Bellman optimality operator $T_{\mu}$. Then, an optimal policy can be obtained by using $J_{\mu}^*(x)$ and the Bellman optimality equation (\ref{optim}). We can also establish the same result by using $Q$-functions instead of value functions. Indeed, let us define the optimal $Q$-function as
$$
Q_{\mu}^*(x,a) = c(x,a,\mu) + \beta \int_{\sX} J_{\mu}^*(y) \, p(dy|x,a,\mu). 
$$
 Since $Q_{\mu,\min}^* = J_{\mu}^*$, we can re-write the equation above 
$$
Q_{\mu}^*(x,a) = c(x,a,\mu) + \beta \int_{\sX} Q_{\mu,\min}^*(y) \, p(dy|x,a,\mu)  \eqqcolon H_{\mu}Q_{\mu}^*(x,a), 
$$
where $H_{\mu}$ is the Bellman optimality operator for $Q$-functions. One can prove that $H_{\mu}$ is a contraction with modulus $\beta\alpha$ and the unique fixed point of $H_{\mu}$ is $Q^*$. Hence, we can develop a value iteration algorithm to compute $Q^*$, and using $Q^*$ we can obtain the optimal policy. The advantage of this algorithm is that one can adapt this algorithm to the model-free setting via $Q$-learning. Therefore, in the remainder of this paper, we will develop value iteration algorithms using $Q$-functions instead value functions.

Let us first define the set on which the $Q$-functions live: 
\begin{align}
&{\cal C} \coloneqq \bigg\{Q:\sX \times \sA \rightarrow [0,\infty); \, \|Q\|_{w} \leq \frac{M}{1-\beta \, \alpha} \text{ } \text{and} \text{ } \|Q_{\min}\|_{\Lip} \leq \frac{L_2}{1-\beta \, K_2} \bigg\},\nonumber 
\end{align}
where $\beta \in (0,1)$ is the discount factor. For any $(x,Q,\mu) \in \sX \times {\cal C} \times \P(\sX)$, by Assumption~\ref{as1}-(e), there exists a unique minimizer $f(x,Q,\mu)$ of
$$
c(x,a,\mu) + \beta \, \int_{\sX} Q_{\min}(y) \, p(dy|x,a,\mu) = F(x,Q_{\min},\mu,a).
$$
Moreover, this unique minimizer $f(x,Q,\mu)$ makes the gradient of $F(x,Q_{\min},\mu,a)$ (with respect to $a$) zero; that is,
$$
\nabla \,F(x,Q_{\min},\mu,f(x,Q,\mu)) = 0.  
$$

Now, we define the mean-field equilibrium (MFE) operator as follows:
$$H: {\cal C} \times \P(\sX) \ni (Q,\mu) \mapsto \left(H_1(Q,\mu),H_2(Q,\mu)\right) \in {\cal C} \times \P(\sX),$$ 
where
\begin{align}
H_1(Q,\mu)(x,a) &\coloneqq c(x,a,\mu) + \beta \, \int_{\sX} Q_{\min}(y) \, p(dy|x,a,\mu) \nonumber \\
H_2(Q,\mu)(\cdot) &\coloneqq \int_{\sX \times \sA} \hspace{-10pt} p(\cdot|x,f(x,Q,\mu),\mu) \, \mu(dx). \nonumber
\end{align}
Here, $f(x,Q,\mu) \in \sA$ is the unique minimizer of
\begin{align}
H_1(Q,\mu)(x,a) \coloneqq c(x,a,\mu) + \beta \, \int_{\sX} Q_{\min}(y) \, p(dy|x,a,\mu).\nonumber
\end{align}
We first prove that $H$ is well defined. 

\begin{lemma}
$H$ maps ${\cal C} \times \P(\sX)$ into itself.
\end{lemma}

\begin{proof}
It is clear that $H_2(Q,\mu) \in \P(\sX)$. Hence, we need to prove that $H_1(Q,\mu) \in {\cal C}$. Let $(Q,\mu) \in {\cal C} \times \P(\sX)$. Then, we have
\begin{align}
\sup_{x,a} \frac{\left| H_1(Q,\mu)(x,a) \right|} {w(x,a)} = &\sup_{x,a} \frac{\left| c(x,a,\mu) + \beta \, \int_{\sX} Q_{\min}(y) \, p(dy|x,a,\mu) \right|} {w(x,a)} \nonumber \\ 
&\leq \sup_{x,a} \frac{\left| c(x,a,\mu) \right|} {w(x,a)} + \beta \, \sup_{x,a} \frac{\left|\int_{\sX} Q_{\min}(y) \, p(dy|x,a,\mu) \right|} {w(x,a)} \nonumber \\
&\leq M + \beta \, \|Q_{\min}\|_{w_{\max}} \, \sup_{x,a} \frac{\left|\int_{\sX} w_{\max}(y) \, p(dy|x,a,\mu) \right|} {w(x,a)}\nonumber \\
&\overset{(1)}{\leq} M+ \beta \, \alpha \,  \|Q\|_{w} \nonumber \\
&\leq M + \beta \, \alpha \, \frac{M}{1-\beta \, \alpha} = \frac{M}{1-\beta \, \alpha}, \nonumber
\end{align}
where (1) follows from Assumption~\ref{as1}-(d) and $\|Q_{\min}\|_{w_{\max}} \leq \|Q\|_w$.
Moreover, for any $x, \hat{x} \in \sX$, we have 
\begin{align}
&|H_1(Q,\mu)_{\min}(x) - H_1(Q,\mu)_{\min}(\hat{x})| \nonumber \\
&=\bigg| \min_{a \in \sA} \bigg[c(x,a,\mu) + \beta \, \int_{\sX} Q_{\min}(y) \, p(dy|x,a,\mu)\bigg] \nonumber \\
&\phantom{xxxxxxxxxxxxx}- \min_{a \in \sA} \bigg[c(\hat{x},a,\mu) + \beta \, \int_{\sX} Q_{\min}(y) \, p(dy|\hat{x},a,\mu)\bigg] \bigg| \nonumber \\
&\leq \sup_{a \in \sA} |c(x,a,\mu) - c(\hat{x},a,\mu)| \nonumber \\
&\phantom{xxxxxxxxxxxxxxxxx}+ \beta \, \sup_{a \in \sA} \left|
\int_{\sX} Q_{\min}(y) \, p(dy|x,a,\mu) - \int_{\sX} Q_{\min}(y) \, p(dy|\hat{x},a,\mu) \right| \nonumber \\
&\overset{(1)}{\leq} L_2 \, d_{\sX}(x,\hat{x}) + \beta \, K_2 \, \|Q_{\min}\|_{\Lip} \, d_{\sX}(x,\hat{x}) \nonumber \\
&\leq \frac{L_2}{1-\beta \, K_2} \, d_{\sX}(x,\hat{x}), \nonumber
\end{align}
where (1) follows from Assumption~\ref{as1}-(a),(b). This implies that $H_1(Q,\mu) \in {\cal C}$ which completes the proof.
\end{proof}
 
Hence, the MFE-operator $H$ is well-defined. Our next goal is to prove that $H$ is a contraction operator. Using this result, we will introduce a value iteration algorithm that will give a mean-field equilibrium.  

\begin{theorem}\label{disc-local}
The mapping $H: {\cal C} \times \P(\sX) \rightarrow  {\cal C} \times \P(\sX)$  is a contraction with constant $k$, where $k$ is the constant in Assumption~\ref{as2}. 
\end{theorem}

\begin{proof}
Fix any $(Q,\mu)$ and $(\hQ,\hmu)$ in ${\cal C} \times \P(\sX)$. First, we analyse the distance between $H_1(Q,\mu)$ and $H_1(\hQ,\hmu)$:  
\begin{align}
&\|H_1(Q,\mu) - H_1(\hQ,\hmu)\|_{w} \nonumber \\
&= \sup_{x,a} \frac{\bigg| c(x,a,\mu) + \beta \int_{\sX} Q_{\min}(y) \, p(dy|x,a,\mu) - c(x,a,\hmu) - \beta  \int_{\sX} \hQ_{\min}(y) \, p(dy|x,a,\hmu) \bigg|}{w(x,a)} \nonumber \\
&\leq  \sup_{x,a} \frac{\big| c(x,a,\mu) - c(x,a,\hmu) \big|}{w(x,a)} \nonumber \\
&\phantom{xxxxxxxxxx}+ \beta \sup_{x,a} \frac{\bigg| \int_{\sX} Q_{\min}(y) \, p(dy|x,a,\mu) -\int_{\sX} \hQ_{\min}(y) \, p(dy|x,a,\hmu) \bigg|}{w(x,a)} \nonumber \\
&\overset{(1)}{\leq} L_1 \, W_1(\mu,\hmu) \nonumber \\
&\phantom{xxxxxxxxxx}+ \beta \sup_{x,a} \frac{\bigg| \int_{\sX} Q_{\min}(y) \, p(dy|x,a,\mu) -\int_{\sX} \hQ_{\min}(y) \, p(dy|x,a,\mu) \bigg|}{w(x,a)} \nonumber \\
&\phantom{xxxxxxxxxx}+ \beta \sup_{x,a} \frac{\bigg| \int_{\sX} \hQ_{\min}(y) \, p(dy|x,a,\mu) -\int_{\sX} \hQ_{\min}(y) \, p(dy|x,a,\hmu) \bigg|}{w(x,a)} \nonumber \\
&\overset{(2)}{\leq} L_1 \, W_1(\mu,\hmu) \nonumber \\
&\phantom{xxxxxxxxxx}+ \beta \, \|Q_{\min} - \hQ_{\min}\|_{w_{\max}} \, \sup_{x,a} \frac{\int_{\sX} w_{\max}(y) \, p(dy|x,a,\mu)}{w(x,a)} \nonumber \\
&\phantom{xxxxxxxxxx}+ \beta \, \|\hQ_{\min}\|_{\Lip} \, \sup_{x,a} \frac{W_1(p(\cdot|x,a,\mu),p(\cdot|x,a,\hmu))}{w(x,a)} \nonumber \\
&\overset{(3)}{\leq} L_1 \, W_1(\mu,\hmu) + \beta \, \alpha \, \|Q - \hQ\|_{w} + \beta \frac{L_2}{1-\beta \, K_2} \, K_1 \, W_1(\mu,\hmu) \label{fin} 
\end{align}
where (1) follows from Assumption~\ref{as1}-(a), (2) follows from the fact that $\hQ_{\min} \in \Lip(\sX)$, and (3) follows from Assumption~\ref{as1}-(b),(d).

Next, we consider the distance between $H_2(Q,\mu)$ and $H_2(\hQ,\hmu)$. To that end, we will first make a perturbation analysis to obtain an upper bound on the difference between the unique minimizer $f(x,Q,\mu)$ of $H_1(Q,\mu)(x,a)$ and the unique minimizer $f(x,\hQ,\hmu)$ of $H_1(\hQ,\hmu)(x,a)$, with respect to $a$. Recall the function 
$$
F: \sX \times {\cal C} \times \P(\sX) \times \sA \ni (x,v,\mu,a) \mapsto 
c(x,a,\mu) + \beta \int_{\sX} v(y) \, p(dy|x,a,\mu) \in \R.
$$
Since it is $\rho$-strongly convex by Assumption~\ref{as1}-(e), it satisfies \cite[Lemma 3.2]{HaRa19}
\begin{align}\label{gradient}
\left[ \nabla F(x,v,\mu,a+r) - \nabla F(x,v,\mu,a) \right]^T \cdot r \geq \rho \, \|r\|^2, 
\end{align}
for any $a,r \in \sA$ and for any $x \in \sX$. Let us set 
$$a = f(x,Q,\mu)$$ and $$r = f(y,\hQ,\hmu) - f(x,Q,\mu).$$ As $a = f(x,Q,\mu)$ is the unique minimizer of a strongly convex function $F(x,Q_{\min},\mu,\cdot)$, we have 
$$
\nabla \, F\left(x,Q_{\min},\mu,f(x,Q,\mu)\right)=0. 
$$
The same is true for $a+r = f(y,\hQ,\hmu)$ and $F(y,\hQ_{\min},\hmu,\cdot)$. Therefore, by Assumption~\ref{as1}-(e) and (\ref{gradient}), we have
\begin{align}
-\nabla F(y,\hQ_{\min},\hmu,a)^T \cdot r &= -\nabla F(y,\hQ_{\min},\hmu,a)^T \cdot r + \nabla F(y,\hQ_{\min},\hmu,a+r)^T \cdot r \nonumber \\
&\geq \rho \, \|r\|^2. \label{st1}
\end{align}
Similarly, by Assumption~\ref{as1}-(e), we also have
\begin{align}
-\nabla F(y,\hQ_{\min},\hmu,a)^T \cdot r &= -\nabla F(y,\hQ_{\min},\hmu,a)^T \cdot r + \nabla F(x,Q_{\min},\mu,a)^T \cdot r \nonumber \\
&\leq \|r\| \, \|\nabla F(x,Q_{\min},\mu,a)-\nabla F(y,\hQ_{\min},\hmu,a)\| \nonumber \\
&\leq K_F \, \|r\| \left(d_{\sX}(x,y) + \|Q_{\min}-\hQ_{\min}\|_{w_{\max}} + W_1(\mu,\hmu)\right) \nonumber \\
&\leq K_F \, \|r\| \left(d_{\sX}(x,y) + \|Q-\hQ\|_{w} + W_1(\mu,\hmu)\right). \label{st2}
\end{align}
Combining (\ref{st1}) and (\ref{st2}) yields 
\begin{align}\label{perturbation}
\|f(y,\hQ,\hmu) - f(x,Q,\mu)\| \leq \frac{K_F}{\rho} \,\left(d_{\sX}(x,y) + \|Q-\hQ\|_w + W_1(\mu,\hmu)\right)
\end{align}

Now, we can start analysing the distance between $H_2(Q,\mu)$ and $H_2(\hQ,\hmu)$. To do that we use the dual formulation of the $W_1$ distance. Indeed, we have
\begin{align}
&W_1(H_2(Q,\mu),H_2(\hQ,\hmu)) \nonumber \\
&= \sup_{\|g\|_{\Lip}\leq1} \bigg| \int_{\sX \times \sA} \int_{\sX} g(y) \, p(dy|x,f(x,Q,\mu),\mu) \, \mu(dx) \nonumber \\
&\phantom{xxxxxxxxxxxxxx}- \int_{\sX \times \sA} \int_{\sX} g(y) \, p(dy|x,f(x,\hQ,\hmu),\hmu) \, \hmu(dx) \biggr| \nonumber \\
&\leq \sup_{\|g\|_{\Lip}\leq1} \bigg| \int_{\sX \times \sA} \int_{\sX} g(y) \, p(dy|x,f(x,Q,\mu),\mu) \, \mu(dx) \nonumber \\
&\phantom{xxxxxxxxxxxxxx}- \int_{\sX \times \sA} \int_{\sX} g(y) \, p(dy|x,f(x,\hQ,\hmu),\hmu) \, \mu(dx) \biggr| \nonumber \\
&+ \sup_{\|g\|_{\Lip}\leq1} \bigg| \int_{\sX \times \sA} \int_{\sX} g(y) \, p(dy|x,f(x,\hQ,\hmu),\hmu) \, \mu(dx) \nonumber \\
&\phantom{xxxxxxxxxxxxxx}- \int_{\sX \times \sA} \int_{\sX} g(y) \, p(dy|x,f(x,\hQ,\hmu),\hmu) \, \hmu(dx) \biggr| \nonumber \\
&\overset{(1)}{\leq} \int_{\sX \times \sA} \sup_{\|g\|_{\Lip}\leq1} \bigg| \int_{\sX} g(y) \, p(dy|x,f(x,Q,\mu),\mu) \, \mu(dx) \nonumber \\
&\phantom{xxxxxxxxxxxxxx}- \int_{\sX \times \sA} \int_{\sX} g(y) \, p(dy|x,f(x,\hQ,\hmu),\hmu) \biggr| \, \mu(dx)  \nonumber \\
&+ \left( K_2 + \frac{K_F}{\rho} \right) \, W_1(\mu,\hmu)\nonumber \\
&\leq \int_{\sX \times \sA} W_1\left(p(\cdot|x,f(x,Q,\mu),\mu),p(\cdot|x,f(x,\hQ,\hmu),\hmu)\right) \, \mu(dx) \nonumber \\
&+ \left( K_2 + \frac{K_F}{\rho} \right) \, W_1(\mu,\hmu) \nonumber \\
&\overset{(2)}{\leq} \frac{K_F}{\rho} K_1  \,\left(\|Q-\hQ\|_w + W_1(\mu,\hmu)\right)+ K_1 \, W_1(\mu,\hmu) + \left( K_2 + \frac{K_F}{\rho} \right) \, W_1(\mu,\hmu). \label{sin}
\end{align}
To show that (1) follows from Assumption~\ref{as1}-(b), let us define the Markov transition probability $P:\sX \rightarrow \P(\sX)$ as
$$
P(\cdot|x) \coloneqq p(\cdot|x,f(x,\hQ,\hmu),\hmu).
$$
Note that, for any $x,y \in \sX$, by Assumption~\ref{as1}-(b) and (\ref{perturbation}), we have 
\begin{align}
W_1(P(\cdot|x),P(\cdot|y)) &\leq K_2 \left( d_{\sX}(x,y) + \|f(x,H_1(\hQ,\hmu),\hmu) - f(y,H_1(\hQ,\hmu),\hmu)\| \right) \nonumber \\
&\leq \left( K_2 + \frac{K_F}{\rho} \right) d_{\sX}(x,y). \label{ubound} 
\end{align}
Let $\xi \in \P(\sX\times\sX)$ be the optimal coupling that achieves Wasserstein distance $W_1(\mu,\hmu)$. Similarly, for any $x,y \in \sX$, let $K(\cdot|x,y) \in \P(\sX\times\sX)$ be the optimal coupling that achieves Wasserstein distance $W_1\left(P(\cdot|x),P(\cdot|y)\right)$. Existence of such couplings follow from \cite[Corollary 5.22]{Vil09}. Note that 
\begin{align}
&\sup_{\|g\|_{\Lip}\leq1} \bigg| \int_{\sX \times \sA} \int_{\sX} g(y) \, p(dy|x,f(x,\hQ,\hmu),\hmu) \, \mu(dx) \nonumber \\
&\phantom{xxxxxxxxxxxxxx}- \int_{\sX \times \sA} \int_{\sX} g(y) \, p(dy|x,f(x,\hQ,\hmu),\hmu) \, \hmu(dx) \biggr| \nonumber \\
&= W_1\left(\mu P,\hmu P\right), \nonumber 
\end{align} 
where $\mu P(\cdot) = \int_{\sX} P(\cdot|x) \, \mu(dx)$ and $\hmu P(\cdot) = \int_{\sX} P(\cdot|x) \, \hmu(dx)$. Let us define $\nu(\cdot) = \int_{\sX \times \sX} K(\cdot|x,y) \, \xi(dx,dy)$. Clearly, $\nu$ is a coupling of $\mu P$ and $\hmu P$. Therefore, we have
\begin{align}
W_1\left(\mu P,\hmu P\right) &\leq \int_{\sX\times\sX} d_{\sX}(x,y) \, \nu(dx,dy) \nonumber \\
&= \int_{\sX \times \sX} \int_{\sX \times \sX} d_{\sX}(x,y) \, K(dx,dy|\hat{x},\hat{y}) \, \xi(d\hat{x},d\hat{y}) \nonumber \\
&= \int_{\sX \times \sX} W_1(P(\cdot|\hat{x}),P(\cdot|\hat{y})) \, \xi(d\hat{x},d\hat{y}) \nonumber \\
&\leq  \left( K_2 + \frac{K_F}{\rho} \right) \, \int_{\sX \times \sX} d_{\sX}(\hat{x},\hat{y}) \, \xi(d\hat{x},d\hat{y}) \text{ } \text{(by (\ref{ubound}))} \nonumber \\
&= \left( K_2 + \frac{K_F}{\rho} \right) \, W_1(\mu,\hmu). \nonumber
\end{align}
Hence, (1) follows. Note that (2) follows from (\ref{perturbation}) and Assumption~\ref{as1}-(b). 

Now, the theorem follows by combining (\ref{fin}) and (\ref{sin}).
\end{proof}

Since we have proven that $H$ is a contraction operator, by Banach Fixed Point Theorem, we can  conclude that the following value iteration algorithm converges to the fixed point of $H$. Using the output of this algorithm, we can then easily construct a mean-field equilibrium as stated in the theorem below.

\begin{algorithm}[h!]
\caption{Value Iteration Algorithm}
\label{valueit}
\begin{algorithmic}
\STATE{Start with $(Q_0,\mu_0)$}
\WHILE{$(Q_{n},\mu_{n}) \neq (Q_{n-1},\mu_{n-1})$}
\STATE{
$(Q_{n+1},\mu_{n+1}) = H(Q_{n},\mu_{n})$
}
\ENDWHILE
\RETURN{Fixed-point $(Q_*,\mu_*)$ of $H$}
\end{algorithmic}
\end{algorithm}

\begin{theorem}\label{main1}
Let $(Q_*,\mu_*)$ be the output of the above value iteration algorithm. Construct the policy $\pi_*(a|x) = \delta_{f^*(x)}(a)$, where 
$$f^*(x) = \argmin_{a' \in \sA} Q_*(x,a').$$
Then, the pair $(\pi_*,\mu_*)$ is a mean-field equilibrium.  
\end{theorem}

\begin{proof}
Note that $(Q_*,\mu_*)$ is a fixed point of $H$; that is, 
\begin{align}
Q_*(x,a) &= c(x,a,\mu_*) + \beta \int_{\sX} Q_{*,\min}(y) \, p(dy|x,a,\mu_*) \label{opt1} \\
\mu_{*}(\cdot) &= \int_{\sX} p(\cdot|x,a,\mu_*) \, \pi_*(a|x) \, \mu_{*}(dx). \label{opt2}
\end{align}
Here, (\ref{opt1}) implies that $\pi_* \in \Psi(\mu_*)$ and (\ref{opt2}) implies that $\mu_* \in \Lambda(\pi_*)$. Hence, $(\pi_*,\mu_*)$ is a mean-field equilibrium. 
\end{proof}

\subsection{Average Cost}\label{average}

In this section, we consider the average cost mean-field game. In addition to Assumption~\ref{as1} (except (\ref{disc-drift})), we assume the following conditions. Note that in place of (\ref{disc-drift}), we assume condition (b) below. 

\begin{assumption}\label{as3}
\begin{itemize}
\item [ ]
\item [(a)] There exists a sub-probability measure $\lambda$ on $\sX$ such that
$$p(\,\cdot\,|x,a,\mu) \geq \lambda(\,\cdot\,)$$
for all $x \in \sX$, $a \in \sA$, and $\mu \in \P(\sX)$.
\item [(b)] There exists non-negative constants $\alpha \in (0,1)$ and $b$ such that 
$$
\int_{\sX} w_{\max}(y) \, p(dy|x,a,\mu) \leq \alpha \, w(x,a) + b. 
$$
\item [(c)] The equality below holds: 
$$\hspace{-25pt}\kappa \coloneqq \max\left\{\alpha + \frac{K_F}{\rho}  K_1 , L_1+\frac{L_2}{1-K_2} K_1 + \bigg(\frac{K_F}{\rho} +1\bigg) K_1+ K_2 + \frac{K_F}{\rho}\right\} < 1.$$
\end{itemize}
\end{assumption}

\begin{remark}
Note that without loss of generality we can take $b = \int w d\lambda$. Indeed, if $b \leq \int w d\lambda$, we can replace $b$ with $\int w d\lambda$ without violating the inequality in Assumption~\ref{as3}-(b). Conversely, if $\int w d\lambda < b$, then by first increasing the value of $\alpha$ so that $\lambda(\sX)+\alpha>1$, and then adding a constant $l$ to $w$, where
\begin{align}
l \coloneqq \frac{b-\int w d\lambda}{\lambda(\sX)+\alpha-1}, \nonumber
\end{align}
we obtain $b  < \int w d\lambda$. Then, as before, we can set $b  = \int w d\lambda$ and Assumption~\ref{as3}-(b) now holds for the new $w$ and $\alpha$.
\end{remark}

Note that condition (b) is so-called the `drift inequality' and condition (a) is the so-called `minorization' condition, both of which were used in the literature for studying the ergodicity of Markov chains (see \cite[Theorem 7.3.11 and Proposition 10.2.5]{HeLa99}). These assumptions are quite general for studying the average cost stochastic control problems with unbounded one-stage costs. The minorization condition was also used to study average cost mean-field games with a compact state space \cite[Assumption A.3]{Wie19}. Assumption~\ref{as3}-(a) is true when the transition probability satisfies conditions R1(a) and R1(b) in \cite{HeMoRo91} (see also \cite[Remark 3.3]{HeMoRo91} and references therein for further conditions). For Assumption~\ref{as3}-(b), we refer the reader to the examples in \cite[Section 7.4]{HeLa99} to see under which conditions on the system components Assumption~\ref{as3}-(b) holds.

For any state-measure $\mu$, let us define the optimal value function as
$$
V_{\mu}^*(x) \coloneqq \inf_{\pi \in \Pi} \limsup_{T \rightarrow \infty} \frac{1}{T} E^{\pi}\biggl[ \sum_{t=0}^{T-1} c(x(t),a(t),\mu) \, \bigg| \, x(0) = x \biggr]. 
$$
We define the operator $R_{\mu}$ as 
\begin{align}
R_{\mu}\, u(x) 
&= \min_{a \in \sA} \biggl[ c(x,a,\mu) + \int_{\sX} u(y) \, p(dy|x,a,\mu) - \int_{\sX} u(y) \, \lambda(dy)\biggr] \nonumber \\
&= \min_{a \in \sA} \biggl[ c(x,a,\mu) + \int_{\sX} u(y) \, q(dy|x,a,\mu)\biggr], \nonumber
\end{align}
where $q(\,\cdot\,|x,a,\mu) \coloneqq p(\,\cdot\,|x,a,\mu) - \lambda(\,\cdot\,)$ is a sub-stochastic kernel and $u:\sX \rightarrow \R$ is a continuous function with finite $w_{\max}$-norm. Under Assumption~\ref{as1} and Assumption~\ref{as3}, one can prove that $R_{\mu}$ is a contraction operator with modulus $\alpha \in (0,1)$ \cite[Theorem 3.21]{SaLiYuBook}. Therefore, for each $\mu$, there exists a fixed point $h_{\mu}$ of $R_{\mu}$ by Banach Fixed Point Theorem. Note that if $R_{\mu} \, h_{\mu} = h_{\mu}$, then we have
$$
h_{\mu}(x) + \rho_{\mu} = \min_{a \in \sA} \biggl[ c(x,a,\mu) + \int_{\sX} h_{\mu}(y) \, p(dy|x,a,\mu) \biggr], 
$$
where $\rho_{\mu} = \int_{\sX} h_{\mu}(y) \, \gamma(dy)$. The last equation is called the average cost optimality equation (ACOE) in the literature \cite[Chapter 5]{HeLa96}. Since 
$$
\lim_{n \rightarrow \infty} \frac{E^{\pi}[h_{\mu}(x(n))]}{n} = 0
$$
for all $\pi \in \Pi$ by Assumption~\ref{as3}-(b), we have \cite[Theorem 5.2.4]{HeLa96}
$$
\rho_{\mu} = V_{\mu}^*(x) \text{ } \text{for all} \text{ } x \in \sX.  
$$
Moreover, if the mapping $f^*: \sX \rightarrow \sA$ attains its minimum in ACOE; that is,
\begin{align}
&\min_{a \in \sA} \bigg[c(x,a,\mu) + \int_{\sX} h_{\mu}(y) \, p(dy|x,a,\mu) \bigg] \nonumber \\
&\phantom{xxxx}= \bigg[c(x,f^*(x),\mu) + \int_{\sX} h_{\mu}(y) \, p(dy|x,f^*(x),\mu) \bigg], 
\end{align}
then the policy $\pi^*(a|x) = \delta_{f^*(x)}(a)$ is optimal.

We now introduce the value iteration algorithm. Similar to the discounted cost case, let us define the set on which $Q$-functions live: 
\begin{align}
&{\cal M} \coloneqq \bigg\{Q:\sX \times \sA \rightarrow [0,\infty); \, \|Q\|_{w} \leq \frac{M}{1-\alpha} \text{ } \text{and} \text{ } \|Q_{\min}\|_{\Lip} \leq \frac{L_2}{1-K_2} \bigg\}.\nonumber 
\end{align}
Note that, for any $(x,Q,\mu) \in \sX \times {\cal M} \times \P(\sX)$, by Assumption~\ref{as1}-(e), there exists a unique minimizer $f(x,Q,\mu)$ of
$$
c(x,a,\mu) + \int_{\sX} Q_{\min}(y) \, p(dy|x,a,\mu) \eqqcolon F(x,Q_{\min},\mu,a).
$$
Moreover, this unique minimizer $f(x,Q,\mu)$ makes the gradient of $F(x,Q_{\min},\mu,a)$ (with respect to $a$) zero; that is,
$$
\nabla \,F(x,Q_{\min},\mu,f(x,Q,\mu)) = 0.  
$$

Now, we define the mean-field equilibrium (MFE) operator  as follows:
$$L: {\cal M} \times \P(\sX) \ni (Q,\mu) \mapsto \left(L_1(Q,\mu),L_2(Q,\mu)\right) \in {\cal M} \times \P(\sX),$$
where  
\begin{align}
L_1(Q,\mu)(x,a) &\coloneqq c(x,a,\mu) + \int_{\sX} Q_{\min}(y) \, q(dy|x,a,\mu) \nonumber \\
L_2(Q,\mu)(\cdot) &\coloneqq \int_{\sX \times \sA} \hspace{-10pt} p(\cdot|x,f(x,Q,\mu),\mu) \, \mu(dx). \nonumber
\end{align}
Here, $f(x,Q,\mu) \in \sA$ is the unique minimizer of
\begin{align}
L_1(Q,\mu)(x,a) \coloneqq c(x,a,\mu) + \int_{\sX} Q_{\min}(y) \, q(dy|x,a,\mu).\nonumber
\end{align}
We first prove that $L$ is well-defined. 

\begin{lemma}
$L$ maps ${\cal M} \times \P(\sX)$ into itself.
\end{lemma}

\begin{proof}
It is sufficient to prove $L_1(Q,\mu) \in {\cal M}$. Let $(Q,\mu) \in {\cal M} \times \P(\sX)$. Then, we have
\begin{align}
\sup_{x,a} \frac{\left| L_1(Q,\mu)(x,a) \right|} {w(x,a)} = &\sup_{x,a} \frac{\left| c(x,a,\mu) + \int_{\sX} Q_{\min}(y) \, Q(dy|x,a,\mu) \right|} {w(x,a)} \nonumber \\ 
&\leq \sup_{x,a} \frac{\left| c(x,a,\mu) \right|} {w(x,a)} + \sup_{x,a} \frac{\left|\int_{\sX} Q_{\min}(y) \, q(dy|x,a,\mu) \right|} {w(x,a)} \nonumber \\
&\leq M + \|Q_{\min}\|_{w_{\max}} \, \sup_{x,a} \frac{\left|\int_{\sX} w_{\max}(y) \, q(dy|x,a,\mu) \right|} {w(x,a)}\nonumber \\
&\overset{(1)}{\leq} M+ \alpha \,  \|Q\|_{w} \nonumber \\
&\leq M + \alpha \, \frac{M}{1-\alpha} = \frac{M}{1-\alpha}, \nonumber
\end{align}
where (1) follows from Assumption~\ref{as3}-(b) and $\|Q_{\min}\|_{w_{\max}} \leq \|Q\|_w$.
Moreover, for any $x, \hat{x} \in \sX$, we have 
\begin{align}
&|L_1(Q,\mu)_{\min}(x) - L_1(Q,\mu)_{\min}(\hat{x})| \nonumber \\
&=\bigg| \min_{a \in \sA} \bigg[c(x,a,\mu) + \int_{\sX} Q_{\min}(y) \, q(dy|x,a,\mu)\bigg] \nonumber \\
&\phantom{xxxxxxxxxxxxx}- \min_{a \in \sA} \bigg[c(\hat{x},a,\mu) + \int_{\sX} Q_{\min}(y) \, q(dy|\hat{x},a,\mu)\bigg] \bigg| \nonumber \\
&\leq \sup_{a \in \sA} |c(x,a,\mu) - c(\hat{x},a,\mu)| \nonumber \\
&\phantom{xxxxxxxxxxxxxxxxx}+ \sup_{a \in \sA} \left|
\int_{\sX} Q_{\min}(y) \, p(dy|x,a,\mu) - \int_{\sX} Q_{\min}(y) \, p(dy|\hat{x},a,\mu) \right| \nonumber \\
&\overset{(1)}{\leq} L_2 \, d_{\sX}(x,\hat{x}) + K_2 \, \|Q_{\min}\|_{\Lip} \, d_{\sX}(x,\hat{x}) \nonumber \\
&\leq \frac{L_2}{1-K_2} \, d_{\sX}(x,\hat{x}), \nonumber
\end{align}
where (1) follows from Assumption~\ref{as1}-(a),(b). This implies that $L_1(Q,\mu) \in {\cal M}$.
\end{proof}

This result implies that MFE-operator $L$ is well-defined. Our next goal is to prove that $L$ is a contraction operator. After that, we will introduce a value iteration algorithm which will give a mean-field equilibrium.  

\begin{theorem}\label{aver-local}
The mapping $L: {\cal M} \times \P(\sX) \rightarrow  {\cal M} \times \P(\sX)$  is a contraction with constant $\kappa$, where $\kappa$ is the constant in Assumption~\ref{as3}. 
\end{theorem}

\begin{proof}
The proof is similar to the proof of Theorem~\ref{disc-local}. Indeed, fix any $(Q,\mu)$ and $(\hQ,\hmu)$ in ${\cal M} \times \P(\sX)$. First, we analyse the distance between $L_1(Q,\mu)$ and $L_1(\hQ,\hmu)$:  
\begin{align}
&\|L_1(Q,\mu) - L_1(\hQ,\hmu)\|_{w} \nonumber \\
&= \sup_{x,a} \frac{\bigg| c(x,a,\mu) + \int_{\sX} Q_{\min}(y) \, q(dy|x,a,\mu) - c(x,a,\hmu) - \int_{\sX} \hQ_{\min}(y) \, q(dy|x,a,\hmu) \bigg|}{w(x,a)} \nonumber \\
&\leq  \sup_{x,a} \frac{\big| c(x,a,\mu) - c(x,a,\hmu) \big|}{w(x,a)} \nonumber \\
&\phantom{xxxxxxxxxx}+ \sup_{x,a} \frac{\bigg| \int_{\sX} Q_{\min}(y) \, q(dy|x,a,\mu) -\int_{\sX} \hQ_{\min}(y) \, q(dy|x,a,\hmu) \bigg|}{w(x,a)} \nonumber \\
&\overset{(1)}{\leq} L_1 \, W_1(\mu,\hmu) \nonumber \\
&\phantom{xxxxxxxxxx}+ \sup_{x,a} \frac{\bigg| \int_{\sX} Q_{\min}(y) \, q(dy|x,a,\mu) -\int_{\sX} \hQ_{\min}(y) \, q(dy|x,a,\mu) \bigg|}{w(x,a)} \nonumber \\
&\phantom{xxxxxxxxxx}+ \sup_{x,a} \frac{\bigg| \int_{\sX} \hQ_{\min}(y) \, q(dy|x,a,\mu) -\int_{\sX} \hQ_{\min}(y) \, q(dy|x,a,\hmu) \bigg|}{w(x,a)} \nonumber \\
&\overset{(2)}{\leq} L_1 \, W_1(\mu,\hmu) \nonumber \\
&\phantom{xxxxxxxxxx}+  \|Q_{\min} - \hQ_{\min}\|_{w_{\max}} \, \sup_{x,a} \frac{\int_{\sX} w_{\max}(y) \, q(dy|x,a,\mu)}{w(x,a)} \nonumber \\
&\phantom{xxxxxxxxxx}+ \|\hQ_{\min}\|_{\Lip} \, \sup_{x,a} \frac{W_1(p(\cdot|x,a,\mu),p(\cdot|x,a,\hmu))}{w(x,a)} \nonumber \\
&\overset{(3)}{\leq} L_1 \, W_1(\mu,\hmu) + \alpha \, \|Q - \hQ\|_{w} + \frac{L_2}{1- K_2} \, K_1 \, W_1(\mu,\hmu) \label{av-fin} 
\end{align}
where (1) follows from Assumption~\ref{as1}-(a), (2) follows from the fact that $\hQ_{\min} \in \Lip(\sX)$, and (3) follows from Assumption~\ref{as1}-(b) and Assumption~\ref{as3}-(b).

Next, we consider the distance between $L_2(Q,\mu)$ and $L_2(\hQ,\hmu)$. First of all, by using a similar analysis as in the proof of Theorem~\ref{disc-local}, we can bound the distance between the unique minimizer $f(x,Q,\mu)$ of $L_1(Q,\mu)(x,a)$ and the unique minimizer $f(x,\hQ,\hmu)$ of $L_1(\hQ,\hmu)(x,a)$ as follows:
\begin{align}\label{av-perturbation}
\|f(y,\hQ,\hmu) - f(x,Q,\mu)\| \leq \frac{K_F}{\rho} \,\left(d_{\sX}(x,y) + \|Q-\hQ\|_w + W_1(\mu,\hmu)\right)
\end{align}
Using this and the same analysis in the proof of Theorem~\ref{disc-local}, we can now obtain the following bound on the distance between $L_2(Q,\mu)$ and $L_2(\hQ,\hmu)$:
\begin{align}
&W_1(L_2(Q,\mu),L_2(\hQ,\hmu)) \label{av-sin} \\
&\phantom{xxxxxxxxxxxxxx}\leq \frac{K_F}{\rho} K_1 \,\left(\|Q-\hQ\|_w + W_1(\mu,\hmu)\right) + (K_1 + K_2 + \frac{K_F}{\rho}) \, W_1(\mu,\hmu). \nonumber
\end{align}
Now, theorem follows by combining (\ref{av-fin}) and (\ref{av-sin}).
\end{proof}

Since we have shown that $L$ is a contraction operator, by Banach Fixed Point Theorem, we can  conclude that the following value iteration algorithm converges to the fixed point of $L$. Using the output of this algorithm, we can then easily construct a mean-field equilibrium as stated in the theorem below.

\begin{algorithm}[h!]
\caption{Value Iteration Algorithm}
\label{valueit}
\begin{algorithmic}
\STATE{Start with $(Q_0,\mu_0)$}
\WHILE{$(Q_{n},\mu_{n}) \neq (Q_{n-1},\mu_{n-1})$}
\STATE{
$(Q_{n+1},\mu_{n+1}) = L(Q_{n},\mu_{n})$
}
\ENDWHILE
\RETURN{Fixed-point $(Q_*,\mu_*)$ of $L$}
\end{algorithmic}
\end{algorithm}

\begin{theorem}\label{main1}
Let $(Q_*,\mu_*)$ be the output of the above value iteration algorithm. Construct the policy $\pi_*(a|x) = \delta_{f^*(x)}(a)$, where 
$$f^*(x) = \argmin_{a' \in \sA} Q_*(x,a').$$
Then, the pair $(\pi_*,\mu_*)$ is a mean-field equilibrium.  
\end{theorem}

\begin{proof}
Note that $(Q_*,\mu_*)$ is a fixed point of $L$; that is, 
\begin{align}
Q_*(x,a) &= c(x,a,\mu_*) + \int_{\sX} Q_{*,\min}(y) \, q(dy|x,a,\mu_*) \label{av-opt1} \\
\mu_{*}(\cdot) &= \int_{\sX} p(\cdot|x,a,\mu_*) \, \pi_*(a|x) \, \mu_{*}(dx). \label{av-opt2}
\end{align}
From (\ref{av-opt1}), we obtain 
\begin{align}
Q_{*,\min}(x) + \rho_{*} = \min_{a \in \sA} \bigg[c(x,a,\mu_*) + \int_{\sX} Q_{*,\min}(y) \, p(dy|x,a,\mu_*)\bigg], \nonumber
\end{align}
where $\rho_* = \int_{\sX} Q_{*,\min}(x) \, \lambda(dx)$. Hence, $(f_*,Q_{*,\min},\rho_*)$ solves the ACOE equation. This implies that $\pi_* \in \Psi(\mu_*)$; that is, $\pi_*$ is an optimal policy for $\mu_*$ (see the discussion at the beginning of this section). Moreover, (\ref{av-opt2}) implies that $\mu_* \in \Lambda(\pi_*)$. Hence, $(\pi_*,\mu_*)$ is a mean-field equilibrium. 
\end{proof}

\section{Conclusion}\label{conc}

This paper has established a value iteration algorithm for discrete time mean-field games subject to discounted and average cost criteria. Under certain regularity conditions on systems components, we have proved that the mean-field equilibrium (MFE) operator in the value iteration algorithm is a contraction. We have then used the fixed point of the MFE operator to construct a mean-field equilibrium.

\end{document}